\documentclass[a4paper,11pt,english]{amsart}

\usepackage{amsaddr}

\makeatletter
\renewcommand{\email}[2][]{%
  \ifx\emails\@empty\relax\else{\g@addto@macro\emails{,\space}}\fi%
  \@ifnotempty{#1}{\g@addto@macro\emails{\textrm{(#1)}\space}}%
  \g@addto@macro\emails{#2}%
}
\makeatother

\usepackage[paperwidth=210mm,paperheight=297mm,inner=3cm,outer=3cm,top=2.5cm,bottom=2.6cm]{geometry}

\usepackage[utf8x]{inputenc}

\usepackage{parskip}
\usepackage{amsmath}
\usepackage[pdftex]{graphicx}
\usepackage{amssymb,mathtools}
\usepackage{upgreek}
\usepackage{cleveref}
\usepackage{babel}
\usepackage{caption}
\usepackage{etex}
\usepackage{fancyvrb}
\usepackage{tikz,tikz-3dplot}
\usepackage{pgfplots}
\usetikzlibrary{calc,datavisualization}
\usetikzlibrary{datavisualization.formats.functions}
\usepackage[ruled,vlined]{algorithm2e}

\usepackage{nicefrac,xfrac}

\usepackage{amsthm}
\newtheorem{thm}{Theorem}

\usetikzlibrary{shadings}
\tdplotsetmaincoords{80}{45}
\tdplotsetrotatedcoords{-90}{180}{-90}

\newcommand{\bra}[1]{{\langle{#1}\rvert}}
\newcommand{\braket}[2]{{\langle{#1}\rvert{#2}\rangle}}
\newcommand{\ket}[1]{{\lvert{#1}\rangle}}

\tikzset{surface/.style={draw=blue!70!black, fill=blue!40!white, fill opacity=.6}}

\numberwithin{equation}{section}

\title[Quantum Random Number Generator based on Violations of CHSH-3]{Quantum Random Number Generator based on Violations of the Free CHSH-3 Inequality}

\date{\today}
\author{Don Jean Baptiste Anoman}
\email{don.anoman@unilim.fr}
\address{Univ. Limoges, CNRS, XLIM, UMR\,7252}
\author{Fran\c{c}ois Arnault}
\email{arnault@unilim.fr}
\address{Univ. Limoges, CNRS, XLIM, UMR\,7252}
\author{Simone Naldi}
\email{simone.naldi@unilim.fr}
\address{Univ. Limoges, CNRS, XLIM, UMR\,7252}

\begin{document}

%%%% This is to de-capitalise the title+authors
%\begingroup
%\def\uppercasenonmath#1{} % this disables uppercasing title
%\let\MakeUppercase\relax % this disables uppercasing authors
%\maketitle
%\endgroup

\maketitle
\bigskip
\bigskip
\bigskip
\bigskip
\begin{abstract}

%\section*{Abstract}
We describe a protocol for generating random numbers based on the existence of quantum violations
of a free version of Clauser-Horne-Shimony-Holt inequality for qutrit, namely CHSH-3. Our method uses semidefinite
programming relaxations to compute such violations. In a standard setting the CHSH-3 inequality
involves two separated qutrits and compatible measurement, that is, commuting with each other,
yielding the known quantum bound of $1+\sqrt{11/3} \approx 2.9149$ .
In our framework, $d$-dimensional quantum systems (qudits) where $d$ is not fixed {\it a priori},
and measurement operators possibly not compatible, are allowed. This loss of constraints yields a
higher value for the maximum expectation of CHSH-3. Based on such upper bound on the violation of CHSH-3, we
develop a random number generator with only one party.
Our protocol generates a maximal entropy and its security is based, through self testing arguments, on the attainability of the
maximal violation of the free CHSH-3 for quantum systems.

\end{abstract}
\newpage

\newpage
\section{ Introduction}
Random number generation is a central question in computer science and related domains like cryptography
and information security. One strategy for dealing with this problem is based on intrinsically
random theories, such as quantum physics. A crucial need is to be able to distinguish between pure
randomness generated by the parties and noise that can derive from imperfections of the setting or even
from a malicious adversary.

A strategy deriving from quantum physics is based on Bell inequalities \cite{bell1964einstein}. These are
inequalities involving the expected values (or probabilities) of outcomes of measurements that hold in
classical mechanics but that can be violated in a quantum setting. This is the case for the CHSH-2
inequality \cite{chsh2}, where it is shown that the quantum bound ($2\sqrt{2}$) is higher than the
classical one, which is $2$, see for instance \cite{leTsar, steph}.

For CHSH-2, the violation is related to the non-locality of quantum physics. Indeed, the complete 
description of a quantum system is not only related to its local environment, but can be correlated to 
a very far system, due to entanglement. Non-locality, in addition to the default random character
of quantum physics, is the basis of the random number generator in \cite{pam}, where a protocol is
developed that relies on a two-parties configuration whose security is yielded by the violation of
CHSH-2. Moreover, thanks to the relation between the violation and the output entropy, the protocol
is proven to be device-independent. This means that in a quantum setting, the user can have a guarantee
on the quality of the randomness, without knowledge on the precise states and measurements that have
been performed.

The generation of random numbers can also be obtained with a one-party system, see for instance
\cite{randomKcbs} for the case of a unique qutrit. This protocol is based on quantum contextuality,
that is, on the property that the measurement result of a quantum observable depends on the set of 
compatible observers. In \cite{randomKcbs} non-contextuality is verified by the KCBS inequality 
\cite{kcbs}. The security of the protocol relies on the fact that a violation of the KCBS inequality 
yields a strictly positive entropy. Such an entropy reaches the maximum for the maximum value of the 
violation of the Bell inequality (see \cite[Fig.\,1]{kcbs}).

In this paper we present a protocol for the generation of random numbers which we call Gabriel. It is 
based on intrinsic randomness of quantum formalism which we test by violation of a free variant of the CHSH-3 Bell inequality. Indeed we use state and measures allowing us to get
the  bound $4$ of a free CHSH-3 expression. This is a non algebraic bound of the latter expression while it is the algebraic bound of the initial expression. This value is greater than
the quantum bound which is available in the literature \cite{FCZ}, which is explained by the fact that
we do not impose that the observables commute to each other (between two parties) or {\it a priori} bounds on the dimension of
the Hilbert space they act on. This viewpoint is motivated by the result in \cite{amblard} where the
it is shown how to implement some products of non-commuting observables.  This protocol has the particularity that the quantum configuration reaching the expected value of $4$ for the
free CHSH-3 yields an entropy of $1$ trit for each trit which is produced.

The paper is organized as follows. \Cref{sec:chsh3} contains a reformulation of the CHSH-3 inequality in a free setting, that is, without commutativity and dimensional constraints involving the observables. In \Cref{sec:sdp} we describe an approach based on convex semidefinite relaxations to compute bounds on 
the violation of CHSH-3. We finally describe the protocol in Section \ref{sec:protocol} and prove
its security giving self testing arguments.
\newpage
\section{CHSH-3 inequality and its free variant}
\label{sec:chsh3}

\subsection{Original setting}
\label{ssec:original}
Clauser-Horne-Shimony-Holt inequalities in the original context involved $2$ parties, $2$ measurements per 
party and $2$ outcomes per measurement (compactly named CHSH-2). Further many authors have worked on 
generalizations with many measurements (for instance in \cite{steph}) and possibly different values for 
the number $d$ of outcomes.

In \cite{chsh3} the case with $3$ outcomes is defined. The authors of \cite{adgl} show that the CHSH-3 expression can be rewritten
as in \cite[Eq.\,3]{cglmp}, which we detail below. Both parties can perform two measurements that are represented by random
variables $A_1,A_2$ for Alice and $B_1,B_2$ for Bob. In a multiplicative setting (as in \cite{FCZ}), the outcomes can be represented by the cubic roots
of unity $1, \omega, \omega^2$, with $\omega$ satisfying $\omega^2+\omega+1=0$. 

The corresponding Bell expression is
\begin{equation}
\begin{split}
%I_3= \; &P(A_1=B_1) + P(B_1= \omega A_2) + P(A_2=B_2)+ P(B_2=A_1) \\
%& - P(A_1=\omega^2 B_1) - P(B_1=A_2) - P(A_2=\omega^2 B_2) - P(B_2=\omega^2 A_1) 
I_3 = \; & P(A_1 = B_1) + P(A_2 = \omega^2B_1) + P(A_2 = B_2)+ P(A_1 = B_2) \\
& - P(A_1 = \omega^2 B_1) - P(A_2 = B_1) - P(A_2 = \omega^2 B_2) - P(A_1 = \omega B_2) 
\end{split}
\label{proba}
\end{equation}
or more explicitely
\begin{equation}
\begin{split}
I_3 = &
\,P(1,1 | A_1 B_1) + P(\omega,\omega | A_1 B_1) + P(\omega^2,\omega^2 | A_1 B_1)+P(\omega^2,1 | A_2 B_1) \\
+ & \,P(1,\omega|A_2 B_1) + P(\omega,\omega^2 |A_2 B_1)+P(1,1|A_2 B_2) + P(\omega,\omega|A_2 B_2) \\
+ & \,P(\omega^2,\omega^2|A_2 B_2)+P(1,1|A_1 B_2) + P(\omega,\omega|A_1 B_2)+ P(\omega^2,\omega^2|A_1 B_2) \\
- & \,P(1,\omega|A_1 B_1) - P(\omega,\omega^2|A_1 B_1) - P(\omega^2,1|A_1 B_1)-P(1,1 | A_2 B_1) \\
- & \,P(\omega,\omega|A_2 B_1)-P(\omega^2,\omega^2|A_2 B_1)-P(1,\omega|A_2 B_2)-P(\omega,\omega^2|A_2 B_2) \\
- & \,P(\omega^2,1|A_2 B_2)-P(\omega,1 | A_1 B_2) - P(\omega^2,\omega| A_1 B_2)- P(1,\omega^2|A_1 B_2)
\end{split}
\end{equation}
where $P(\omega^k,\omega^\ell|A_iB_j)$ denotes the probability of getting $\omega^k,\omega^\ell$ with
measurements $A_i,B_j$. The classical bound of $2$ is satisfied in a local realistic setting
\cite{cglmp} and establishes what one generally calls the {\it CHSH-3 inequality}: $I_3 \leq 2$.

In a quantum setting, $A_1,A_2,B_1$ and $B_2$ are observables acting on a three-dimensional Hilbert space 
$\mathbb{H}$ with eigenvalues $1, \omega, \omega^2$ defined as above. The corresponding eigenvectors are
denoted by $\ket{a_{i,1}}, \ket{a_{i,\omega}}, \ket{a_{i,\omega^2}}$ for $A_i$, $i=1,2$, and similarly 
for $B_1,B_2$. This allows us to define the projectors 
$$
\begin{array}{rrrr}
A_{1,1} = \ket{a_{1,1}}\bra{a_{1,1}} & A_{2,1} = \ket{a_{2,1}}\bra{a_{2,1}} & B_{1,1} = \ket{b_{1,1}}\bra{b_{1,1}} & B_{2,1} = \ket{b_{2,1}}\bra{b_{2,1}} \\
A_{1,\omega} = \ket{a_{1,\omega}}\bra{a_{1,\omega}} & A_{2,\omega} = \ket{a_{2,\omega}}\bra{a_{2,\omega}} & B_{1,\omega} = \ket{b_{1,\omega}}\bra{b_{1,\omega}} & B_{2,\omega} = \ket{b_{2,\omega}}\bra{b_{2,\omega}} \\
A_{1,\omega^2} = \ket{a_{1,\omega^2}}\bra{a_{1,\omega^2}} & A_{2,\omega^2} = \ket{a_{2,\omega^2}}\bra{a_{2,\omega^2}} & B_{1,\omega^2} = \ket{b_{1,\omega^2}}\bra{b_{1,\omega^2}} & B_{2,\omega^2} = \ket{b_{2,\omega^2}}\bra{b_{2,\omega^2}} \\
\end{array}
$$
and the corresponding decomposition for $A_1$ (similarly for $A_2,B_1,B_2$):
$$
A_1 = 1 \cdot \ket{a_{1,1}}\bra{a_{1,1}} + \omega \cdot \ket{a_{1,\omega}}\bra{a_{1,\omega}}
+ \omega^2 \cdot \ket{a_{1,\omega^2}}\bra{a_{1,\omega^2}}
$$
Under the assumption that the observables $A_i$s commute with the $B_j$s, the following equality 
holds:
$$
\bra{\phi} A_{i,\omega^k} B_{j,\omega^\ell} \ket{\phi} = P(\omega^k,\omega^\ell | A_{i}B_{j})
$$
for a state $\ket{\phi} \in \mathbb{H}$, $i,j \in \{1,2\}$ and $k,\ell \in \{0,1,2\}$. Thus one can rewrite the expression in \Cref{proba} as
\begin{equation}
\label{el-proj}
\begin{split}
%%%%%% I_3 ( &\mid \phi \rangle, A_1,A_2,B_1,B_2)=  \\  
 & \bra{\phi} A_{1,1}B_{1,1} + A_{1,1} B_{2,1} - A_{1,1} B_{1,\omega} - A_{1,1} B_{2,\omega^2} + A_{1,\omega}B_{1, \omega} + A_{1, \omega} B_{2, \omega} \\ 
      & -A_{1, \omega} B_{1, \omega^2} - A_{1,\omega} B_{2,1} +A_{1, \omega^2} B_{1, \omega^2} + A_{1, \omega^2} B_{2, \omega^2} - A_{1, \omega^2} B_{1,1} - A_{1, \omega^2} B_{2,\omega} \\
& + A_{2,1} B_{1,\omega} +  A_{2,1} B_{2,1}- A_{2,1} B_{1,1} -  A_{2,1} B_{2,\omega} + A_{2, \omega} B_{2, \omega} + A_{2,\omega} B_{1,\omega^2}   \\
& - A_{2, \omega} B_{1,\omega} - A_{2,\omega} B_{2, \omega^2} + A_{2, \omega^2} B_{1,1} + A_{2, \omega^2} B_{2, \omega^2} - A_{2, \omega^2} B_{1, \omega^2} - A_{2, \omega^2} B_{2,1} \ket{\phi}  \\
\end{split}
\end{equation}

In this case (commutative observables) we recall that the quantum bound for $I_3$ is $1 + \sqrt{{11}/{3}} 
\approx 2.9149$, see \cite{FCZ}, yielding a violation of $(\sqrt{{11}/{3}}+1) /2 \approx 1.4574$ for the CHSH-3.
In this paper, we use a semidefinite-programming-based strategy to compute upper bounds on the violation
of a special version of CHSH-3, which is described below in \Cref{ssec:freechsh3}.

\subsection{Free CHSH-3 inequality}
\label{ssec:freechsh3}
Let us describe the precise setting we are working on. Our goal is to consider a non-commutative version
of \Cref{el-proj}, and where the dimension of the Hilbert space the observables are operating on,
is not fixed {\it a priori}. That is we are interested in a {\it free} CHSH-3 inequality.

Whereas the standard setting consists of two parties (Alice and Bob) with four given observables,
two for each party ($A_1,A_2,B_1,B_2$ as previously discussed in \Cref{ssec:original}), our model consists
of one single party with four observables $X_1, X_2, X_3, X_4$, acting on states $\ket{\phi}$ living
in a Hilbert space $\mathbb{H}$ of unconstrained dimension. 

The observables $X_i$ are possibly not commuting to each other, they are unknown and will be explicitely 
constructed by solving a single semidefinite program, the details are given in \Cref{sec:sdp}. For each
$i \in \{1,2,3,4\}$, and $j \in \{0,1,2\}$, as in \Cref{ssec:original} we decompose each $X_i$ as follows:
$$
X_i = 1 \cdot X_{i,1} + \omega \cdot X_{i,\omega} + \omega^2 \cdot X_{i,\omega^2}, 
\,\,\,\, \text{ for } i \in \{1,2,3,4\}
$$
introducing $12$ variables $X_{i,\omega^k}, i \in \{1,2,3,4\}, k \in \{0,1,2\}$ corresponding to the
projector $\ket{x_{i,\omega^k}}\bra{x_{i,\omega^k}}$ on the eigenvector $\ket{x_{i,\omega^k}}$ of the $X_i'$s
(see \cite[Sec.\,2.2]{QCQI}). 

Therefore the CHSH-3 quadratic form can be formally restated as function of $X = (X_{1,1},X_{1,\omega},
\ldots,X_{4,\omega^2})$ and of the state $\ket{\phi}$ as \bigskip
\begin{equation}
\label{el-proj2}
\begin{split}
%%%%%% I_3 ( &\mid \phi \rangle, A_1,A_2,B_1,B_2)=  \\ 
& \bra{\phi} f(X) \ket{\phi} \quad \quad \text{ with} \\
& f(X) = \\
 & = X_{1,1}X_{3,1} + X_{1,1} X_{4,1} - X_{1,1} X_{3,\omega} - X_{1,1} X_{4,\omega^2} + X_{1,\omega}X_{3, \omega} + X_{1, \omega} X_{4, \omega} \\ 
      & -X_{1, \omega} X_{3, \omega^2} - X_{1,\omega} X_{4,1} +X_{1, \omega^2} X_{3, \omega^2} + X_{1, \omega^2} X_{4, \omega^2} - X_{1, \omega^2} X_{3,1} - X_{1, \omega^2} X_{4,\omega} \\
& + X_{2,1} X_{3,\omega} +  X_{2,1} X_{4,1}- X_{2,1} X_{3,1} -  X_{2,1} X_{4,\omega} + X_{2, \omega} X_{4, \omega} + X_{2,\omega} X_{3,\omega^2}   \\
& - X_{2, \omega} X_{3,\omega} - X_{2,\omega} X_{4, \omega^2} + X_{2, \omega^2} X_{3,1} + X_{2, \omega^2} X_{4, \omega^2} - X_{2, \omega^2} X_{3, \omega^2} - X_{2, \omega^2} X_{4,1} \\
\end{split}
\end{equation}
where the previous products are non-commutative. Thus remark that for non-commutative operators,
$\bra{\phi} X_{i,\omega^k} X_{j,\omega^\ell} \ket{\phi}$ does not in general correspond to 
$P(\omega^k,\omega^\ell|X_{i}X_{j})$.

%\begin{equation}
%\begin{split}
% I_{3 \; extended} ( &\mid \phi \rangle, C_1, C_2, C_3, C_4)=  \\  
%& \langle \phi \mid [  C_{1,1} C_{3,1} + C_{1,1} C_{4,1} - C_{1,1} C_{3,\omega} - C_{1,1} C_{4,\omega^2} \\ 
%&+C_{1, \omega} C_{3, \omega} + C_{1, \omega} C_{4, \omega} -  C_{1, \omega} C_{3, \omega^2} - C_{1,\omega} C_{4,1} \\
%&+C_{1, \omega^2} C_{3, \omega^2} + C_{1, \omega^2} C_{4, \omega^2} - C_{1, \omega^2} C_{3,1} - C_{1, \omega^2} C_{4,\omega} \\
%&+ C_{2,1} C_{3,\omega} +  C_{2,1} C_{4,1}- C_{2,1} C_{3,1} -  C_{2,1} C_{4,\omega}   \\
%&+ C_{2, \omega} C_{4, \omega} + C_{2,\omega} C_{3,\omega^2} - C_{2, \omega} C_{3,\omega} - C_{2,\omega} C_{4, \omega^2} \\
%& + C_{2, \omega^2} C_{3,1} + C_{2, \omega^2} C_{4, \omega^2} - C_{2, \omega^2} C_{3, \omega^2} - C_{2, \omega^2} C_{4,1}
%] \mid \phi \rangle  \\
% & \text{ without  commutativity and dimension constraints}
%\end{split} 
%\end{equation}

Let us also mention that since \Cref{el-proj2} reduces to \Cref{proba} assuming commutativity, one thereby
deduces the classical bound $f(X) \leq 2$ in a local realistic model. One cannot directly derive a quantum 
bound from results in the literature. In \Cref{sec:sdp} we construct explicit non-commutative operators 
$X_i$ yielding a gap of $2$ with respect to the classical bound.

\section{Explicit non algebraic violation of the free CHSH-3 inequality}
\label{sec:sdp}

%\textbf{Cette partie vise} à expliquer les différentes étapes de l'optimisation, dans le cadre d'observables quantiques, menée sur l'extension de CHSH-3 (\ref{el-proj2}). Cette extension ,obtenue en supprimant les contraintes de commutativité et de dimensions des observables,  atteint une espérance de 3.1547 (nous le verrons par la suite, dans ce même chapitre). Nous explicitons de plus, l'état et les mesures permettant d'atteindre cette valeur, cela en utilisant la méthode décrite dans la preuve du théorème 2 de \cite{piro}.   \\

\subsection{Semidefinite relaxations}
Semidefinite Programming (SDP for short) is a class of convex optimization problems that has gained
momentum in the last years. It is a natural generalization of linear programming consisting of the
minimization of linear functions over affine sections of the cone of positive semidefinite symmetric 
matrices. As for linear programs, efficient implementations of the interior-point method are available
in solvers such as \cite{SeDuMi,anderson2012mosek}.

SDP is a versatile tool that is used for {solving} non-convex polynomial optimization problems, that is,
for minimizing multivariate polynomial functions over sets defined by polynomial inequalities 
\cite{Anjos2012}. In \cite{lasserre01globaloptimization} Lasserre defined a hierarchy of SDP problems 
that can be constructed from the original one, and whose minima form an increasing sequence of lower 
bounds of the original optimal value, with asymptotic convergence. Under further conditions on the 
rank of the optimal matrices along the relaxation, the hierarchy converges in finite time to the sought 
solution and the minimizers can be extracted essentially by performing linear algebra operations 
\cite{LasserreSemialg}.

The SDP hierarchy has been extended to the non-commutative setting \cite{burgdorf2016optimization} and 
successfully applied to quantum information, see \cite{PiroSdp} and \cite[Ch.\,21]{Anjos2012}. 
The hierarchy in \cite{PiroSdp} allows
one to get bounds on the minimum or maximum of the action of a non-commutative polynomial function 
of observables, possibly subject to equalities and inequalities. The key idea of such a hierarchy is 
to linearize the quantity $\langle \phi, f(X) \phi \rangle$ where $f(X) = \sum_w f_w w(X)$ is a non-commutative polynomial function of
$n$ measurement operators $X = (X_1,\ldots,X_n)$ defined on a Hilbert space $\mathbb{H}$, $w(X)$ 
is a monomial on $X$, and $ \ket{\phi} \in \mathbb{H}$ is a pure state. The linearization consists of replacing 
the action $\langle \phi, w(X) \phi \rangle$ of the monomial $w(X)$ on the state $ \ket{\phi}$, with a new 
variable, or {\it moment}, $y_w$. In other words, one replaces the original non-linear operator on $X$ 
with the following linear function on the space of variables $y$:
\[
\langle \phi, f(X) \phi \rangle = \sum_w f_w \langle \phi, w(X) \phi \rangle = \sum_w f_w y_w.
\]
The moments $y_w$ up to some order $d$ are then organized in a symmetric multi-hankel {\it moment matrix}
$M_d(y) = (y_{vw})_{v,w}$ (that is, the entry of $M_d(y)$ indexed by $(v,w)$ is $y_{vw}$).
By construction of $y_w$, one gets the necessary condition that $M_d(y)$ is positive semidefinite, 
from the fact that $z^* M_d(y) z \geq 0$ for any complex vector $z = (z_w)$. Similarly, non-linear
constraints can be linearized and lead to additional linear and semidefinite constraints on variables
$X$ in the relaxation.
%$$
%z^* M_d(y) z = \sum_{\alpha,\beta} z_\alpha M_d(y)(\alpha,\beta) z_\beta = 
%$$

In the case of the CHSH-2 inequality for two space-like separated parties, many measurements settings with
two outcomes, the first level of the hierarchy is sufficient to compute Tsirelson's bounds \cite{steph}.
In this work, we use semidefinite programming in the spirit of \cite{steph,PiroSdp} to compute explicit
(non-commuting) observables yielding a violation of the CHSH-3 inequality higher than the known value of
$1+\sqrt{11/3}$.

\subsection{First relaxation of the free CHSH-3}
\label{ssec:firstrelax}
Let $X = (X_{1,1},X_{1,\omega},\ldots,X_{4,\omega^2})$ be the (unknown) projectors on the eigenstates of
operators $X_1,X_2,X_3,X_4$ related to eigenvalues $1, \omega, \omega^2$, as defined in \Cref{sec:chsh3}, 
and let $f(X)$ be the non-commutative quadratic polynomial defined in \Cref{el-proj2}. Since our goal is 
to compute the maximal violation of CHSH-3 with no dimensional constraints, we let $k_4,k_5,\ldots,k_d$ 
be the additional eigenvalues up to dimension $d$ (see for instance \cite[\S 2.2.6]{QCQI}) and similarly
we denote by $X_{i,k_j}$ the projectors onto the eigenstate corresponding to $k_j$, $j\in \{4,\ldots,d\}$.

Let us introduce the following compact notation for the indices of $X_{i,\mu}$.
We define the set
$
T = \{(i,\mu) \mid i = 1,2,3,4, \mu = 1,\omega,\omega^2,k_4,\ldots,k_d\}.
$
Hence the variables $X_{i,\mu}$ are exactly those of the form $X_\alpha$ with $\alpha=(i,\mu) 
\in T$ for some $i,\mu$. Thus the original problem can be stated as follows:
\begin{equation}
\label{orig_opt}
\begin{array}{rcll}
f^* := & \sup        & \bra{\phi} f(X) \ket{\phi}                     & \\
       & \text{s.t.} & \braket{\phi}{\phi} = 1                        & \\
%       &             & X_{i,\mu}X_{i,\nu} = \delta_{\mu\nu} X_{i,\mu} & \text{ for } i \in \{1,2,3,4\}  \\ 
       &             & X_{\alpha}X_{\beta} = \delta_{\mu\nu} X_{\alpha} & \text{ for } \alpha = (i,\mu), \beta = (i,\nu) \in T  \\ 
%       &             & \sum_\mu X_{i,\mu} = 1                         & \text{ for } i \in \{1,2,3,4\}
       &             & \sum_\mu X_{\alpha} = 1                         & \text{ for } i \in \{1,2,3,4\}, \text{ where } \alpha = (i,\mu)
\end{array}
\end{equation}
where $\delta_{\mu\nu}$ is the Kronecker delta for indices $\mu,\nu \in \{1,\omega,\omega^2,k_4,\ldots,
k_d\}$. The two last constraints are related to the equality $X_{i,\mu} = \ket{x_{i,\mu}}\bra{x_{i,\mu}}$ 
that we want to impose, as discussed above.

We denote by $y_\alpha = \bra{\phi} X_\alpha \ket{\phi}$ for $\alpha \in T$, the moment of order one
associated to the variable $X_\alpha$ and to state $\ket{\phi}$ (omitted in the notation). Similarly 
we denote by $y_{\alpha\beta} = \bra{\phi} X_\alpha X_\beta \ket{\phi}$ the moments of order two.
Note that we have $X_{\alpha} X_{\beta} = (X_{\beta}X_{\alpha})^\dagger  $,  because $X_\alpha$ are projectors (hence Hermitian). Therefore the expected values are conjugated each over.

The first moment relaxation of \Cref{orig_opt} is thus expressed in the following form
\begin{equation}
\label{mom_rel}
\begin{array}{rcll}
f_1^* := & \sup & \sum_{\alpha} c_\alpha y_\alpha & \\
         & \text{s.t.} & y_0 = 1                  & \\
& & y_{\alpha\beta} = \delta_{\mu\nu} y_{\alpha} & \text{ for } \alpha = (i,\mu), \beta = (i,\nu) \in T, i \in \{1,2,3,4\} \\
& & \sum_\mu y_{\alpha} = 1 & \text{ for } i \in \{1,2,3,4\}, \text{ where } \alpha = (i,\mu) \\
& & M_1(y) \succeq 0 & 
\end{array}
\end{equation}
where $c_\alpha \in \{-1, 0, 1\}$ are such that $f(X) = \sum_{\alpha} c_\alpha X_\alpha$, and
$M_1(y)$ is the moment matrix of order $1$, namely the matrix
$$
M_1(y) =
\bra{\phi}
v_1 v_1^{\dagger}
\ket{\phi}
=
\left[
\begin{array}{ccccc}
y_0          & y_{\alpha_1}         & y_{\alpha_2}         & \cdots & y_{\alpha_{4d}} \\
y_{\alpha_1} & y_{\alpha_1\alpha_1}       & y_{\alpha_1\alpha_2} & \cdots & y_{\alpha_1\alpha_{4d}} \\
y_{\alpha_2} & y_{\alpha_1\alpha_2} & y_{\alpha_2\alpha_2}       &        & \vdots \\
\vdots       & & & & \\
y_{\alpha_{4d}} & \cdots & & & y_{\alpha_{4d}\alpha_{4d}}
\end{array}
\right]
$$
Above we have chosen an order for indices $\alpha$ in $T = \{\alpha_1,\ldots,\alpha_{4d}\}$,
and denoted the vector of moments up to degree $1$ by $v_1 = (1, X_{\alpha_1}, X_{\alpha_2}, 
\cdots, X_{\alpha_{4d}}) \in \mathbb{C}^{4d+1}$. Problem \eqref{mom_rel} is a relaxation of 
Problem \ref{orig_opt} which implies that $f^* \leq f_1^*$.

For two symmetric matrices $C_1,C_2$, we denote by $C_1 \bullet C_2 = \text{Trace}(C_1 C_2)$ 
the usual Euclidean inner product. Let $C,A_0,A_{\alpha\beta},A_i$ be the $(1+4d) \times (1+4d)$ 
symmetric matrices such that 
$\sum_\alpha c_\alpha y_\alpha = C \bullet M_1(y)$, 
$y_0 = C_0 \bullet M_1(y)$,
$y_{\alpha\beta} - \delta_{\mu\nu} y_{\alpha} = A_{\alpha\beta} \bullet M_1(y)$ and
$\sum_\mu y_\alpha = A_i \bullet M_1(y)$. Thus the problem in \Cref{mom_rel} is equivalent to 
the semidefinite program
\begin{equation}
\label{mom_rel_sdp}
\begin{array}{rcll}
f_1^* := & \sup & C \bullet M_1(y) & \\
         & \text{s.t.} & C_0 \bullet M_1(y) = 1                  & \\
& & A_{\alpha\beta} \bullet M_1(y) = 0 & \text{ for } \alpha = (i,\mu), \beta = (i,\nu) \in T, i \in \{1,2,3,4\} \\
& & A_i \bullet M_1(y) = 1 & \text{ for } i \in \{1,2,3,4\}, \text{ where } \alpha = (i,\mu) \\
& & M_1(y) \succeq 0. &
\end{array}
\end{equation}

%\noindent Avec $W $ la matrice des coefficients de l'espérance $ I_{3 \, extended}$; et $M_1 $  la matrice de la première relaxation définie comme $ M_1 = \begin{pmatrix}
%& & & & \\
%& & \langle \phi \mid P_{h, \epsilon} P_{J,\eta} \mid \phi \rangle && \\
%& & & & \\
%\end{pmatrix}  $  avec   $ P_{h, \epsilon}  = A_{ij} $ ou $ P_{h, \epsilon}  = B_{ij} $ un projecteur.
%\\ $ $ \\

Solving this SDP with SeDuMi \cite{SeDuMi} gives a value of 
$$
f_1^* = 4.
$$

We remark that this value is not the algebraic bound of  the expression \eqref{el-proj2} while it is the one of \eqref{proba} (see \Cref{non algebraic}).

The $13 \times 13$ submatrix $M^*$ of the optimal moment matrix $M_1(y^*)$, corresponding to variables 
$X$ occurring in the CHSH-3 inequality, has the following form:
\begin{equation}
  \label{mom_mat}
  M^* = 
\frac{1}{9}
  \scalebox{0.85}{ %% si jamais, on peut rendre plus petit
    $\left[
  \begin{array}{rrrrrrrrrrrrr}
  9 & ${3}$ & ${3}$ & ${3}$ & ${3}$ & ${3}$ & ${3}$ & ${3}$ & ${3}$ & ${3}$ & ${3}$ & ${3}$ & ${3}$ \\
  ${3}$ 	& ${3}$ & $0$ & $0$ & $0$ & $3$ & $0$ & $2$ & $-1$ & $2$ & $2$ & $ 2$ & $-1$ \\
  ${3}$ & $0$ & $3$ & $0$ & $0$ & $0$ & $3$ & $2$ & $2$ & $-1$ & $-1$ & $2$ & $2$ \\
  ${3}$ & $0$ & $0$ & $3$ & $3$ & $0$ & $0$ & $-1$ & $2$ & $2$ & $ 2$ & $-1$ & $ 2$ \\
  ${3}$ & $0$ & $0$ & $3$ & ${3}$ & $0$ & $0$ & $-1$ & $2$ & $2$ & $2$ & $-1$ & $2$ \\
  ${3}$ & $3$ & $0$ & $0$ & $0$ & $3$ & $0$ & $ 2$& $-1$ & $2$ & $2$ & $2$ & $-1$ \\
  ${3}$ & $0$ & $3$ & $0$ & $0$ & $0$ & ${3}$ & $2$ & $2$ & $-1$ & $-1$ & $ 2$ & $2 $ \\
  ${3}$ & $2$ & $2$ & $-1$ & $-1$& $2$ & $2$  & $3$ & $0$ & $0$ & $0$ & $3$ & $0$ \\
  ${3}$ & $-1$ & $ 2$ & $ 2$& $2 $ & $-1$ & $2$ & $0$ & $3$ & $0$ & $0$ & $0$ & ${3}$ \\
  ${3}$ & $ 2$ & $-1$ & $2$ & $2$ & $2$ & $-1$ & $0$ & $0$ & ${3}$ & ${3}$ & $0$ & $0$ \\
  ${3}$ & $2$ & $-1$ & $2$ & $2$ & $2$ & $-1$ & $0$ & $0$ & ${3}$ & ${3}$ & $0$ & $0$   \\
  ${3}$ &$2$ & $2$ & $-1$ & $-1$ & $2$ & $2$ & ${3}$ & $0$ & $0$ & $0$ & ${3}$ & $0$ \\
  ${3}$ & $-1$ & $2$ & $2$ & $2$ & $-1$ & $2$ & $0$ & ${3}$ & $0$ & $0$ & $0$ & ${3}$
  \end{array}
  \right]$}
\end{equation}
The matrix $M^*$ has rank three and it is positive 
semidefinite, with eigenvalues $\frac{4}{3},\frac{7}{3}$ and $0$ of multiplicity $2,1$ and $10$,
respectively.

In order to retrieve the optimal projectors, we thus compute a factorization of $M^*$ of the form $M^* = B^TB$ (certifying that $M^* \succeq 0$), with $B$ the following $3 \times 13$ matrix

\begin{equation}
\label{matrixB}
\frac{\sqrt{3	}}{9}
  \scalebox{0.8}{
    $\left[
      \begin{array}{rrrrrrrrrrrrr}
        $3$ & 3 & 0 & 0 & 0 & 3 & 0 & 2 & -1 & 2 & 2 & 2 & -1 \\
        $3$ & 0 & 3 & 0 & 0 & 0 & 3 & 2 & 2 & -1 & -1 & 2 & 2\\
        $3$ & 0 & 0 & 3 & 3 & 0& 0 & -1 & 2 & 2 & 2 & -1 & 2 \\
      \end{array}
      \right]$
  }
\end{equation}

As in \cite[Ch.\,21]{Anjos2012}, the first column of $B$ is interpreted as the optimal state $\ket{\phi^*}$,
and for $i \in \{1,2,3,4\}$, the normalization of columns $3i-1,3i$ and $3i+1$ of $B$ as the eigenstates
$\ket{x_{i,1}},\ket{x_{i,\omega}}$ and $\ket{x_{i,\omega^2}}$ corresponding to projective measurements
$X^*_i$ that can be recovered as in \cite[\S 2.2.6]{QCQI}, as follows:
\begin{equation}
\label{operators}
X^*_i = 1 \cdot \ket{x_{i,1}}\bra{x_{i,1}} + \omega \cdot \ket{x_{i,\omega}}\bra{x_{i,\omega}}
+ \omega^2 \cdot \ket{x_{i,\omega^2}}\bra{x_{i,\omega^2}} . 
\end{equation}
We thus have :

\begin{equation}
\begin{array}{cl}
X_1^*= Z= \scalebox{0.8}{
    $ \left[
      \begin{array}{rrr}
        1 & 0& 0  \\
        0 & \omega & 0 \\
        0 & 0 & \omega^2  \\
      \end{array}
      \right]$
  } &  X_2^*= \scalebox{0.8}{
    $ \left[
      \begin{array}{rrr}
        \omega & 0& 0  \\
        0 & \omega^2 & 0 \\
        0 & 0 & 1  \\
      \end{array}
      \right]$
  }  \\ \\
  X_3^*= \frac{1}{3} \scalebox{0.8}{
    $  \left[
      \begin{array}{rrr}
        - \omega  & 2 & 2 \omega^2  \\
        2 & - \omega^2 & 2 \omega \\
       2 \omega^2 & 2 \omega & -1  \\
      \end{array}
      \right]$
  } & X_4^*=\frac{1}{3} \scalebox{0.8}{
    $  \left[
      \begin{array}{rrr}
        - \omega^2 & 2 \omega & 2  \\
        2 \omega & -1 & 2 \omega^2 \\
        2 & 2 \omega^2 & -\omega  \\
      \end{array}
      \right]$
  } 
\end{array}
\end{equation}

We prove the following result concerning the relaxation in \Cref{mom_rel_sdp}.
\begin{thm}
\label{main}
The optimal value of Problem \eqref{orig_opt} is $4$ and it is attained for the
configuration in \Cref{operators} and for $\ket{\phi^*}=(1/ \sqrt{3}) (1,1,1)^\dagger$.
\end{thm}
\begin{proof}
First, we remark that the operators constructed in \Cref{operators} satisfy the constraints in
Problem \eqref{orig_opt}, which yields 
$$
4 = C \bullet M_1(y^*) = \bra{\phi^*} f(X^*) \ket{\phi^*} \leq f^*.
$$ 
Moreover \Cref{mom_rel} is a relaxation of \Cref{orig_opt}, hence, the feasible set in \Cref{mom_rel} contains that of \Cref{orig_opt}  that is, $f^* \leq f_1^* = 4 $,
and we conclude.
%%%and that the evaluation of the objective function on $M_1(y^*)$ yields
\end{proof}

\subsection{Algebraic and non algebraic bound of the free CHSH-3 inequality} $ $ \\
\label{non algebraic}
The initial expression \eqref{proba} of CHSH-3 and that of free CHSH-3 \eqref{el-proj2} are not equivalent, precisely when it comes to non commuting observables simply because \eqref{proba}  does not exist. Thus, they are not supposed to have the same algebraic bound. \\
Moreover, considering the free CHSH-3 expression  \eqref{el-proj2}, we have 24 terms of the form $ \bra{\phi} X_{i, \omega^\ell} X_{j ,\omega^k} \ket{\phi}; $   with absolute value of $1$ $\quad (i, \, j \in \{1,...,4 \}, \; k, \, l \in \{0,1,2\} )$.  Therefore, the algebraic bound of the free CHSH-3 expression is $24$. 

\section{ The protocol}
\label{sec:protocol}

First of all, we present the properties of the states and measurements previously deduced from the optimisation. This properties allows us to design a randomness generator.
We close by showing the security of this protocol based on the retrieval of the algebraic value of CHSH-3 expectation.

\subsection{Measurements and states properties }
\label{property}
Since the measurements $X_1^* $ and $X_2^* $ (the same for $X_3^* $ and $X_4^* $) commutes, the next five facts are equivalents : 
\begin{itemize}
\item[a)] Use the measurement whose matrix is $X_1^* X_2^* = X_2^* X_1^*$. 
\item[b)] First measure $ X_2^*$ and then measure $ X_1^*$. 
\item[c)] First measure $ X_1^*$ and then measure $ X_2^*$. 
\item[d)] Only measure $X_1^*$. 
\item[e)] Only measure $X_2^*$. 
\end{itemize} 

The equivalence holds because 
$$ \ket{x_{1, \omega^k}} = \ket{x_{2, \omega^{k+1}}} = \ket{(x_1x_2)_{\omega^{2k+1}}}$$
 where $ \ket{(x_1x_2)_{\omega^{2k+1}}}$ is the eigenvector of $X_1^*X_2^* $ related to the eigenvalue $ \omega^{2k+1} $ and $ k$, $ k+1$ $2k+1 $ are taken modulo $3$ .
 The difference only relies on the outcome returned by the measurement : 
$$ \begin{array}{lcl}
 X_1 \mapsto \omega^k  & &  X_1X_2 \mapsto \omega^k \omega^{k+1}  \\
 
 X_2 \mapsto \omega^{k+1}  & &  X_2 \quad then \quad  X_1 \mapsto \omega^k \; or \; \omega^{2k+1}
 \end{array}$$

The previous reasoning can be applied to $ X_3 X_4$.  \\ 
We also have to notice that measurements $X_1 $ and $ X_2$ do not commute with measurements $X_3 $ and $ X_4$.
\bigskip
\subsubsection*{Randomness generation}
We also notice that the measurement of the state $ \ket{ \phi^*} $ in the base $ X_i$,  $ i \in \{ 1, 2, 3, 4\} $ yields : 
$$
\begin{array}{rclc}
1 & \text{ with probability } & P(1 \mid (\ket{ \phi^*}, X^*_i)) = |\braket{\phi^*}{x_{i,1}}|^2 & = 1/3\\
\omega & \text{ with probability } & P(\omega \mid (\ket{ \phi^*},X^*_i)) = |\braket{\phi^*}{x_{i,\omega}}|^2 & = 1/3 \\ 
\omega^2 & \text{ with probability } & P(\omega^2 \mid (\ket{ \phi^*},X^*_i)) = |\braket{\phi^*}{x_{i,\omega^2}}|^2 & = 1/3 \\
\end{array}
$$
In that respect, a randomness generator based on this states and measurements will have quality given by the min-entropy of (see \cite{konig2009operational})
$$
H_{\infty} = -\log_3 \, \max_{\ell , i} \, P(\omega^\ell \mid (\ket{ \phi^*}, X^*_i)) = -\log_3 \, 1/3 = 1
$$
concluding that the min-entropy for each trit is thus equal to $1$ trit.
We use this fact to construct the following protocol. \\

\subsection{Protocol execution}
\label{ssec:Proto}
Let's consider the state and measurements $ \ket{ \phi^*} , \, X_1^*, X_2^*, X_3^*, \allowbreak X_4^*$ as defined in \Cref{ssec:firstrelax}. We assume that a public source of random numbers is available, such as that of NIST\footnote{{https://csrc.nist.gov/projects/interoperable-randomness-beacons}}.
Based on previous discussion, let us now describe how our protocol works in practice to generate a random trit. The following steps are iterated:
\begin{enumerate}
\item
  The user uniformly choose a random couple  of measurements $( X_i^*,X_j^* )$, $ i,j \in \{ 1,2,3,4 \} $ .
\item
  If $ i,j \in \{ 1,2\} $ or $ i,j \in \{ 3,4\} $ then we apply the measurement $X_i^* $ to the state $ \ket{\phi^*} $. The outcome is returned as \textbf{random trit}.
\item
  Otherwise, we do the measurement $X_j$ on the state $ \ket{\phi^*}  $. We collect the resulting state $ \ket{x_{j, \omega^k}} $, \; $k \in \{ 0,1,2 \} $. This state is then measured in the basis $ X_i$ and the resulting state $ \ket{x_{i, \omega^\ell}} $, \; $\ell \in \{ 0,1,2 \} $.
  This outcomes are stored and used to evaluate the expectation of free CHSH-3 as explained in \Cref{vioBell}.
\end{enumerate}

\subsection{Maximal value of Bell expectation}
\label{ssec:secu}

Here, we can evaluate the expectation of CHSH-3 $\bra{\phi} f(X) \ket{\phi}$ on the optimal configuration $(X^*,\ket{\phi^*})$ computed 
in \Cref{sec:sdp} through the SDP relaxation, getting the maximum violation of $4$. This implies that one can detect the potential
interference of an eavesdropper if such violation is not attained.
%Valeur maximale d'espérance de l'inégalité de Bell
\label{vioBell}
To evaluate this free CHSH-3 expectation, we use the outcomes of our protocol described in section \ref{ssec:Proto}.
The maximum free CHSH-3 expectation  is attained for the configuration (\ref{matrixB}) that gives moments  $ y_{(i, \omega^\ell)(j, \omega^k)}^*  $ that can be experimentally estimated. In fact we have for any moment $ y_{(i, \omega^\ell)(j, \omega^k)} $ : 

\begin{equation*}
\begin{split}
y_{(i, \omega^\ell)(j, \omega^k)} &= \bra{\phi} X_{i, \omega^\ell} X_{j ,\omega^k} \ket{\phi}  \\
  &= \braket{\phi}{x_{i, \omega^\ell}} \braket{x_{i, \omega^\ell}}{x_{j, \omega^k}} \braket{x_{j, \omega^k}}{ \phi}
\end{split}
\end{equation*}

Thus, we got
\begin{equation*}
\begin{split}
\mid y_{(i, \omega^\ell)(j, \omega^k)}  \mid &= \sqrt{ \mid \braket{\phi}{x_{i, \omega^\ell}} \mid ^2 } \;  \sqrt{ \mid \braket{x_{i, \omega^\ell}}{x_{j, \omega^k}} \mid^2 } \\ & \quad \; \sqrt{ \mid \braket{x_{j, \omega^k}}{ \phi}  \mid ^2}
\end{split}
\end{equation*}

where $\mid \langle \phi \mid x_{i, \omega^{ \ell}}  \rangle \mid ^2 =\bra{\phi} X_{i,\omega^{ \ell}} \ket{\phi} =  P( \ket{x_{i, \omega^\ell}} \mid( \ket{\phi} , X_i) ) $ is the probability  to retrieve the state $ \ket{x_{i, \omega^{ \ell}}} $ when it's about to measure the state $ \ket{\phi} $ in the basis $ X_i$.  We thus have : 

\begin{equation*}
\begin{split}
\mid y_{(i, \omega^\ell)(j, \omega^k)}\mid &= \sqrt{ P( \ket{x_{i, \omega^\ell}} \mid( \ket{\phi} , X_i) ) }  \\\ 
& \quad \times \sqrt{ P( \ket{x_{j, \omega^k}} \mid( \ket{x_{i, \omega^\ell}} , X_j) ) } \\ & \quad \; \times \sqrt{ P( \ket{x_{j, \omega^k}} \mid( \ket{\phi} , X_j) )}
\end{split}
\end{equation*}

The good value can be retrieved with the relation 
$$ \sum_k y_{(i, \omega^\ell)(j, \omega^k)} =   \sum_k \bra{\phi} X_{i, \omega^\ell} X_{j ,\omega^k} \ket{\phi}  = \bra{\phi} X_{i, \omega^\ell} \ket{\phi} = P( \ket{x_{i, \omega^\ell}} \mid( \ket{\phi} , X_i) )     $$

\noindent In so doing, we can experimentally evaluate this moment. 
Indeed the probabilities involved in previous expression can be deduced from the different outcomes of the step three of our protocol. We can thus compute the Bell expectation of the protocol. Therefore, if the state and measurements of the protocol are those presented according to the table (\ref{matrixB}), we get the maximal Bell expectation $4$. \\

 In the following section we will show what can be
deduced in the case of maximal value of the Bell expression.

\bigskip
\subsection{ \textbf{Self Testing arguments}} $ $ \\

We want to give a witness of the quality of the generated randomness, depending on the Bell expectation of the experiment, providing the fact that the device is honest but error prone. In fact, we can exhibit a relationship between the Bell expectation and the lower bound of the quantity of randomness produced under quantum assumption. This relation, using the outcome statistics, helps us to estimate the quality of the generated numbers. \\
To do so, we proceed by the way explained in \cite{randomKcbs} and adapted to our context. Here, we give a lower bound of the min-entropy as a function of the Bell Expectation :  For a given configuration $ ( X_1, ... , X_4, \ket{\phi}) $ the min entropy is given by 

$$ H_{min} ( X_1, ... , X_4, \ket{\phi}) = \log_3 \, \max_{\ell,i}   P(\omega^\ell \mid (\ket{ \phi},X_i)) $$

where $ \ell \in \{0,1,2\}, \, i \in \{ 1,2,3,4 \} \;$,  $ P(\omega^\ell \mid (\ket{ \phi},X_i)) = \bra{ \phi} X_{i,\ell} \ket{\phi } $

We want to find a lower bound of the min-entropy for a given free CHSH-3 expectation $ L $. This bound must hold for any configuration reaching this Bell value $ L $. This is equivalent to solve the following problem

\begin{equation}
\label{certi-eq}
\begin{array}{rlll}
\max_{i,\ell} \;  &          \bra{\phi} X_{i,\ell} \ket{\phi} & \\
\text {s.t.}       & \bra{\phi} f(X) \ket{\phi}= L             & \\
 & \text{ the same constraints as in }  \eqref{orig_opt} & \\
\end{array}
\end{equation} 

where $X=(X_1,...,X_4) $. We solve them according to same method as in \eqref{orig_opt}. In practice, for this value $L $ we optimize each moment of order $1$. And then, we take the maximum of this values : " max of max ". Doing it repeatedly for different values of $ L$, the following curve is obtained:  \\ 
\begin{center}
\begin{table} [h]

\includegraphics[scale=0.21]{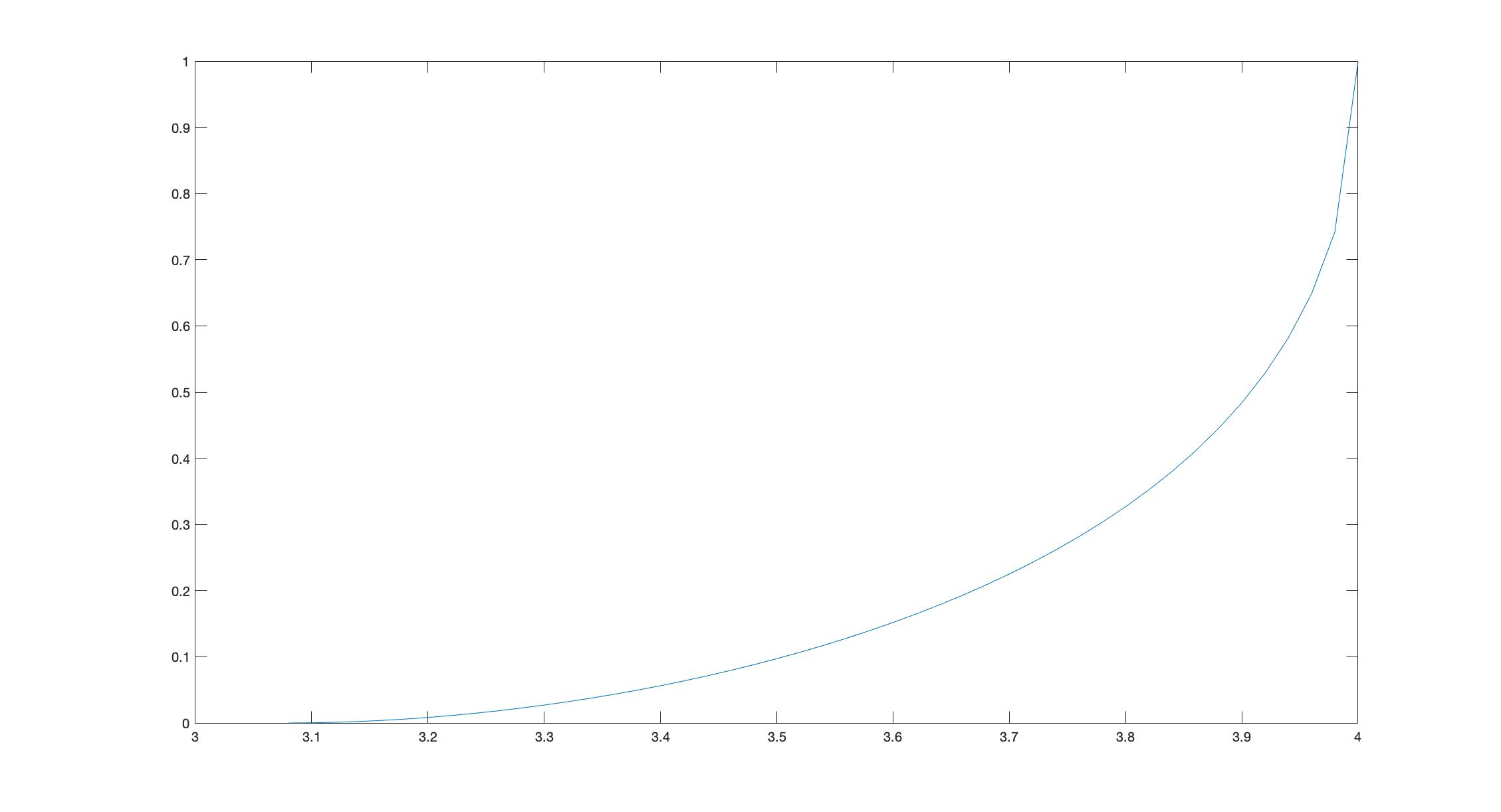}
\caption{ The min-entropy bound f(L) versus different levels L of the Free CHSH-3 violation}
\label{certif}
\end{table}

\end{center}

The previous curve, as in \cite[Figure 1]{randomKcbs}, reaches the maximum entropy only for the highest Bell value $4$. Unlike the previous reference, when the value $L$ is greater than the classical CHSH-3 bound $2$, the min entropy remains null. This until the value $ L>3.08$. \\
In our context, under quantum conditions, the figure \ref{certif} shows us that obtaining the maximum Bell value $4$ is equivalent to obtain a min-entropy  $ H_{min} = 1\, trit$ for each random trit produced.  \\

\bigskip

A further work could be to estimate values $\epsilon$  leading to a valid protocol when expectation $4-\epsilon $ is reached.

\section{Conclusion}

\noindent
To conclude, we have shown a quantum randomness generator GABRIEL based on the observables and states deduced from the optimisation of a free version of CHSH-3. In our case we removed the constraints of commutativity and dimension of observables. It yields a greater bound, $4 $, than the one in commutativity an dimension constrained case $ 1+ \sqrt{11/3} $.

We thus deduces from the optimum matrix of moments the state and measurements, which can generate randomness. The reliability of this protocol relied on the fact that we have the min-entropy as a function of the free CHSH-3 expectation \Cref{certif}.  From this, we see that, reaching the maximal Bell expectation is equivalent to have the maximal entropy $ H_{min} =1 $ for each trit produced.

%\bibliography{biblio}
%\bibliographystyle{abbrv}

\bigskip
\bigskip

\end{document}